%% file: LinkPrediction.tex
\documentclass[sigconf]{aamas}  

\usepackage{amsmath}
\usepackage{amsthm}
\usepackage{amsfonts}
\usepackage{algorithm}
\usepackage{algpseudocode}
\usepackage{multirow}
\usepackage{graphicx}
\usepackage{subcaption}

\begin{document}

\title{Attacking Similarity-Based Link Prediction in Social Networks}  
\author{Kai Zhou$^1$, Tomasz P. Michalak$^2$, Talal Rahwan$^3$, Marcin Waniek$^{2,3}$, and Yevgeniy Vorobeychik$^1$\\
	\small{$^1$Department of Computer Science and Engineering, Washington University in St. Louis, St. Louis\\
	$^2$Institute of Informatics, University of Warsaw, Warsaw, Poland\\
	$^3$Department of Computer Science, Khalifa University of Science and Technology, Abu Dhabi, UAE}}

\begin{abstract}  
Link prediction is one of the fundamental problems in computational
social science.
A particularly common means to predict existence of unobserved links
is via structural similarity metrics, such as the number of common
neighbors; node pairs with higher similarity are thus deemed more
likely to be linked.
However, a number of applications of link prediction, such as
predicting links in gang or terrorist networks, are adversarial,
with another party incentivized to minimize its effectiveness by
manipulating observed information about the network.
We offer a comprehensive algorithmic investigation of the problem of
attacking similarity-based link prediction through link deletion, focusing on two broad
classes of such approaches, one which uses only local information
about target links, and another which uses global network information.
While we show several variations of the general problem to be NP-Hard
for both local and global metrics, we exhibit a number of
well-motivated special cases which are tractable.
Additionally, we provide principled and empirically effective
algorithms for the intractable cases, in some cases proving worst-case
approximation guarantees.
\end{abstract}

%


\maketitle


\input{intro}

\input{related_work}
\input{prob}
\input{local}

\input{global}

\input{exp}

\section{Conclusion}
We investigate the problem of hiding a set of target links in a network via minimizing the similarities of those links,  by deleting a limited number of edges. We divide similarity metrics associated with potential links into two broad classes: local metrics (CND and WCN) and global metrics (Katz and ACT). We prove that computing optimal attacks on all these metrics is NP-hard.

For local metrics, we proposed an algorithm minimizing the upper bounds of local metrics, which corresponds to maximizing submodular functions under cardinality constraints. Furthermore, we identify two special cases, attacking a single link and attacking a group of nodes, where the first case ensures optimal attacks for all local metrics and the latter ensures optimal attacks for CND metrics. For global metrics, we prove that even when attacking a single link, both the problem of minimizing Katz and that of maximizing ACT are NP-Hard. We then propose an efficient
greedy algorithm (Greedy-Katz) and a principled heuristic algorithm
(Local-ACT) for the two problems, respectively. Our experiments show
that our algorithms are highly effective
in practice and, in particular, significantly outperform a recently
proposed heuristic. Overall, the results in this paper greatly advance the algorithmic understanding of attacking similarity-based link prediction.
	
\bibliographystyle{ACM-Reference-Format}  
\bibliography{multilink_citation}  

\clearpage
\end{document}

%% file: intro.tex
\section{Introduction}

Link prediction is a fundamental problem in social network analysis.
A common approach to predicting a target link $(u,v)$ is to use an observed (sub)network to infer the likelihood of the existence of this link using a measure of \emph{similarity}, or closeness, of $u$ and $v$; we call this \emph{similarity-based link prediction}~\cite{liben2007link,wang2015link,al2006link,zhang2018link}.
For example, if $u$ and $v$ are individuals who have many friends in common, it may be natural to assume that they are themselves friends. Representational power of social networks implies very broad application of link prediction techniques, ranging from friend recommendations to inference of criminal and terrorist ties.

A crucial assumption in conventional similarity-based link prediction approaches is that the observed (sub)network is measured correctly. However, insofar as link prediction may reveal relationships which associated parties prefer to keep hidden---either for the sake of privacy, or to avoid being apprehended by law enforcement---it introduces incentives to manipulate network measurements in order to reduce perceived similarity scores for target links.

In order to systematically study the ability of an ``adversary'' to manipulate link prediction, we formulate attacks on link prediction as an optimization problem in which the adversary aims to minimize the total weighted similarity scores of a set of target links by removing a limited subset of edges from the observed subnetwork.
We present a comprehensive study of this algorithmic problem, focusing on two important subclasses of similarity metrics: \emph{local metrics}, which make use of only local information about the target link, and \emph{global metrics}, which use global network information.
We show that the problem is in general NP-Hard even for local metrics, and our hardness results are stronger for the commonly used Katz and ACT global similarity metrics (for example, the problem is hard for these metrics even if there is only a single target link).

On the positive side, we exhibit a number of important special cases when the problem is tractable.
These include attacks on local metrics when there is a single target link, or a collection of target \emph{nodes} (such as gang members) with the goal of hiding links among them.
Additionally, we present practical algorithms for the intractable cases, including global similarity metrics.
In a number of such settings, we are able to provide provable approximation guarantees.
Finally, we demonstrate the effectiveness of the approaches we develop through an extensive experimental evaluation.

%% file: related_work.tex
\paragraph{Related Work}
Link prediction has been extensively studied in multiple domains such as social science \cite{liben2007link}, bioinformatics \cite{almansoori2012link}, and security \cite{huang2009time}. 
There are two broad classes of approaches for link prediction: the first based on structural similarity~\cite{liben2007link,lu2011link} and the second using learning~\cite{menon2011link,Wang17,al2006link,Wagn18}.
This work is focused on the former, which commonly use either local information~\cite{leicht2006vertex,zhou2009predicting}, rely on paths between nodes~\cite{katz1953new,lu2009similarity}, or make use of random walks~\cite{fouss2007random} (we view the latter two categories as examples of global metrics).

Our work is connected to several efforts studying vulnerability of social network analysis (SNA). Michalak et al. \cite{michalak2017strategic} suggest considering strategic considerations in SNA, but do not offer algorithmic analysis.
Waniek et al. study attacks against centrality measures and community detection~\cite{waniek2017construction, waniek2018hiding}.
There is considerable literature on hiding or anonymizing links on networks~(e.g., \cite{zhang2016measuring, waniek2018attack,Yu18}), but these approaches allow arbitrary graph modifications and are in any case heuristic, often proposing randomly swapping or rerouting edges.
In contrast, we provide the first comprehensive algorithmic study of the problem of hiding links by merely \emph{deleting} observed edges (i.e., preventing them from being observed), and the first strong positive algorithmic results.

%% file: prob.tex
\section{Problem Formulation}

\subsection{Similarity Metrics}

One of the major approaches for link prediction both in the network science literature and in practice is via the use of similarity metrics~\cite{liben2007link}.
Specifically, suppose we wish to know whether a particular link $(u,v)$ connecting nodes $u$ and $v$ exists.
A structural similarity metric $\mathsf{Sim}(u,v)$ quantifies the extent to which the nodes $u$ and $v$ have shared topological properties, such as shared neighbors, with the idea that higher similarity scores imply greater likelihood that $u$ and $v$ are connected.
Below, we will distinguish two types of similarity metrics: \emph{local}, which only use information about the nodes and their immediate neighbors, and \emph{global}, which make use of global information about the network.

\subsection{Attack Model}

At the high level, our goal is to remove a subset of observed edges in order to minimize perceived similarity scores of a collection of target (and, presumably, existing) links.
This could be viewed both from the perspective of vulnerability analysis, where the goal of link prediction is to identify relationships among malicious parties (such as gang members), or privacy, where the ``attacker'' is not malicious, but rather aims to preserve privacy of a collection of target relationships.

To formalize the problem, consider an underlying graph $\mathcal{G}=(V,E)$ representing a social network, where $V$ is the set of nodes and $E$ is the set of edges.
This graph is not fully known, and instead an analyst obtains answers for a collection of edge queries $Q$ from the environment, where for each query $(u,v) \in Q$, they observe the associated edge if $(u,v) \in E$, and determine that the edge doesn't exist otherwise.
The partially constructed graph $\mathcal{G}_Q = (V_Q,E_Q)$ based on the queries $Q$ is then used to compute similarity metrics $\mathsf{Sim}(u',v')$ for any potential edges $(u',v') \notin Q$.

An attacker has a collection of target links $H$ they wish to hide, and can remove a subset of at most $k$ edgs in $E_Q \equiv E \cap Q$ to this end.
While there are many ways to express the attacker's objective mathematically, a relatively natural and general approach is to minimize the weighted sum of similarity scores of links in $H$:

\begin{align}
	\label{OPT-1}
	\min_{E_a \subset E_Q} \  f_t(E_a) \equiv \sum_{(u,v) \in H} w_{uv}\mathsf{Sim}(u,v;E_a),\quad \text{s.t.}\ |E_a| \leq k,
\end{align}
where $w_{uv}$ is the weight representing the relative importance of hiding the link $(u,v)$, and we make explicit the dependence of similarity metrics on the set of removed  edges $E_a$.
Henceforth, we simplify notation by keeping this dependence implicit.

%% file: local.tex
\section{Attacking Local Similarity Metrics}
Our analysis covers nine representative local similarity metrics (summarized in the supplement) that are commonly used in the state-of-the-art link prediction algorithms.  We first systematically divide local metrics into two sub-class: Common Neighbor Degree (CND) and Weighted Common Neighbor (WCN) metrics, depending on their special structures. Next, we show that attacking all local metrics is NP-Hard. We follow this negative result with an approximation algorithm exhibiting a solution-dependent bound. 
Finally, we present polynomial-time algorithms for well-motivated special cases.

We begin by introducing some notation. We denote $U =\{u_i\}$ as the union of end-nodes, termed \emph{target nodes}, of the target links in $H$. Assume $|U|= n$. Let $W =\{w_1,w_2,\cdots, w_m\}$ be the set of common neighbors of the target nodes, where each $w_i \in W$ connects to \emph{at least} two nodes in $U$.  Let $N(u_i,u_j)$ denote the set of common neighbors of $u_i$ and $u_j$. For any node $u_i \in V$, let $d(u_i)$ be its degree. We use a \emph{decision matrix} $X \in \{0,1\}^{m\times n}$ to denote the states of edges among the nodes in $W$ and $U$, where the entry $x_{ij}$ in the $i$-th row and $j$-th column of $X$ equals $1$ if there is an edge between $w_i$ and $u_j$; otherwise, $x_{ij} =0$.  We will say the attacker \emph{erases} $x_{ij}$ (when $x_{ij} = 1$) to denote the fact that the attacker deletes the edge between $w_i$ and $u_j$ (thus setting $x_{ij}$ as $0$).

\subsection{Classification of Local Metrics}

We now make a useful distinction between two classes of local metrics that use somewhat different local information.

\begin{definition}
	\label{def-CND}
	A metric $\mathsf{Sim}$ is a CND metric if the corresponding total similarity $f_t$ has the form $\sum_{r=1}^m W_r\frac{\sum_{i,j|(u_i,u_j)\in H}x_{ri}\cdot x_{rj}}{f_r(S_r)}$, where $f_r$ is a metric-dependent increasing function of $S_r$, the sum of $r$th row of decision matrix $X$, and $W_r$ is an associated weight.
\end{definition}
The metrics Adamic-Adar (AA), Resource Allocation (RA), and Common Neighbors (CN) are CND metrics. We note that the sum $\sum_{i,j|(u_i,u_j)\in H}x_{ri}\cdot x_{rj}$ is over all links in $H$. For simplicity, we write the sum as $\sum_{ij}$ henceforth.

\begin{definition}
	A metric $\mathsf{Sim}$ is a WCN metric if
	\begin{itemize}
		\item it has the form $\mathsf{Sim}(u_i,u_j) = \frac{|N(u_i,u_j)|}{g(d(u_i),d(u_j),|N(u_i,u_j)|)}$, where $g$ is \emph{strictly} increasing in $d(u_i)$ and $d(u_j)$. That is $g(d(u_i)-t,d(u_j)-s) \leq g(d(u_i),d(u_j))$ for any valid non-negative integers $t$ and $s$ and any valid value of $|N(u_i,u_j)|$.
		\item $\mathsf{Sim}$ is \emph{strictly} increasing in $|N(u_i,u_j)|$. That is, $\mathsf{Sim}(|N(u_i,u_j)| - t) \leq \mathsf{Sim}(|N(u_i,u_j)|)$, for any valid non-negative integer $t$  and any valid values of $d(u_i)$ and $d(u_j)$.	
	\end{itemize}	
\end{definition}

The WCN metrics include many common metrics, such as Jaccard, S\o rensen, Salton, Hub Promoted, Hub Depressed, and Leicht.

By the above definitions, we know a rational attacker will only delete edges between nodes in $W$ and nodes in $U$, since deleting other types of edges will either decrease $d(u_i)$ or $d(w_i)$, causing the similarity to increase. Thus, the total similarity $f_t$ is fully captured by the decision matrix $X$. As a result, attacking local similarities is formulated as an optimization problem, termed as \emph{Prob-Local}:

\begin{align}
\label{OPT-local}
\min_{X} \  f_t(X),\quad \text{s.t.}\ \mathsf{Sum}(X^0 - X) \leq k,
\end{align}
where $X^0$ is the original decision matrix and $\mathsf{Sum}(\cdot)$ denotes the element-wise summation.

\subsection{Hardness Results}
We start by making no restrictions on the set of target links $H$.  In this general case, we show  that attacking all local metrics is NP-hard.
\begin{theorem}
	\label{thm-general}
	Attacking local similarity metrics is NP-Hard.
\end{theorem}
\begin{proof}
As attacking local similarity metrics is modelled as an optimization problem, we consider the corresponding decision problem: can an attacker delete up to $k$ edges such that the total similarity $f_t$ is no greater than a constant $\theta$? We note that the minimum possible $f_{t}$ for all local metrics in a connected graph is $0$. Thus, we consider the decision problem $P_L$, which is to decide whether one can we delete $k$ edges such that $f_{t} = 0$.

We use the vertex cover problem for reduction. Let $P_{VC}$ denote the decision version of vertex cover, which is to decide whether there exists  a vertex cover of size $k$ given a graph $\mathcal{G}$ and an integer $k$.

Given an instance of vertex cover (i.e., a graph $\mathcal{G} = (V,E)$ and an integer $k$), we construct our decision problem $P_L$ as follows. We first construct a new graph $\mathcal{Q}$ in the following steps:
\begin{itemize}
	\item For each node $v_i \in V$, create a node $v_i$ for graph $\mathcal{Q}$.
	\item Add another node $w$ to $\mathcal{Q}$ and connect $w$ to each $v_i$.
	\item Add $n = |V|$ nodes $u_1,\cdots,u_n$ and add an edge between each pair of nodes $(u_i,v_i)$.
	\item Add an edge between $(u_{i},u_{i+1})$, for $i=1,2,\cdots n-1$.
\end{itemize}
The set $H$ of target links  is then $H= \{(v_{i},v_{j})\}$ in $\mathcal{Q}$ if and only if $(v_{i},v_{j})$ is an edge in $\mathcal{G}$. Our decision problem $P_L$ is then constructed regarding this graph $\mathcal{Q}$ and target set $H$.
\begin{figure}[ht]
	\centering
	\includegraphics[scale=0.25]{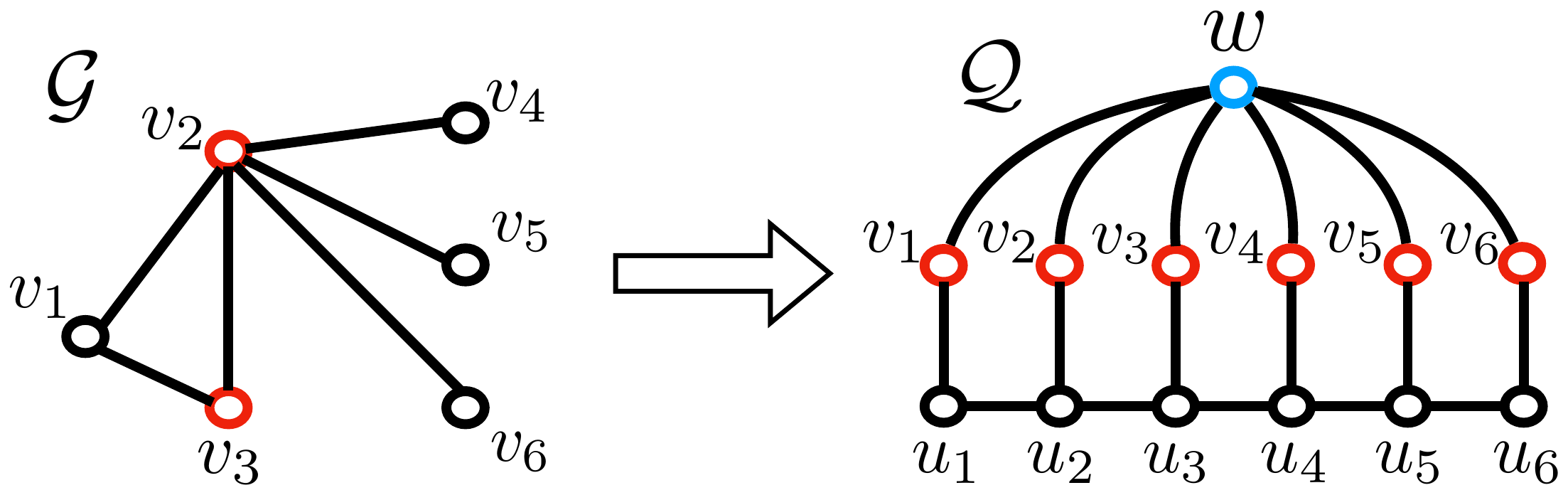}
	\caption{Example: constructing graph $\mathcal{Q}$ from $\mathcal{G}$. $V_c =\{v_2,v_3\}$ is a set cover in $\mathcal{G}$ while $H = \{(v_1,v_2),(v_1,v_3),(v_2,v_3),(v_2,v_4),(v_2,v_5),(v_2,v_6)\}$ is the set of target links in $\mathcal{Q}$.}
	\label{fig-plot-local}
\end{figure}

Now, we show $P_L$ and $P_{VC}$ are equivalent. We use CN metrics as an example and show that the same proof can be applied to other local metrics by slightly modifying the constructed graph $\mathcal{Q}$.

First, we show if there is a vertex cover of size $k$ in graph $\mathcal{G}$, then we can delete $k$ edges such that $f_{t}(H) = 0$ in $\mathcal{Q}$. Suppose $V_{c}$ is a vertex cover with $|V_{c}|=k$. Without loss of generality, let $V_{c}=\{v_{1},\cdots,v_{k}\}$. Then we show that delete $k$ edges $(v_{1},w),\cdots,(v_{k},w)$ will make $f_{t}(H) = 0$. Let $(v_i,v_j) \in H$ be an arbitrary target link. Then $(v_i,v_j)$ corresponds to an edge in $\mathcal{G}$. By the definition of vertex cover, we have at least one of $v_i$ and $v_j$ is in $V_c$. We can assume $v_i \in V_c$. Since $v_{i} and v_{j}$ has only one common neighbor $w$ in $\mathcal{Q}$, deleting $(v_{i},w)$ will make $CN(v_{i},v_{j})=0$. As $(v_i,v_j)$ is arbitrarily selected, we have $CN(v_{i},v_{j})=0$ for any target link $(v_i,v_j) \in H$.  Thus, we have found $k$ edges whose deletion will make $f_t(H)=0$. 

Second, we show if we can delete $k$ edges to make $f_{t}(H)=0$ in $\mathcal{Q}$, the we can find a vertex cover of size $k$ in $\mathcal{G}$. Suppose we found $k$ edges whose deletion will make $f_t(H) = 0$. Then each deleted edge must be $(w,v_i)$ for some $i = 1,\cdots, n$, since deleting other types of edges will not decrease $f_t(H)$. Without loss of generality, we assume the $k$ deleted edges are $(w,v_1),\cdots, (w,v_k)$. We then show that $V_c = \{v_1,\cdots, v_k\}$ forms a vertex cover in $\mathcal{G}$. Since $\forall (v_i,v_j) \in H$, $CN(v_i,v_j) \geq 0$, $f_t(H) = 0$ means that $CN(v_i,v_j) = 0$ for very target link. As each target link $(v_i,v_j)$ initially has one common neighbor $w$, we know at least one of $v_i$ and $v_j$ is in set $V_c$; otherwise, $CN(v_i,v_j) = 1$ making $f_t(H) > 0$. As each $(v_i,v_j)$ corresponds to an edge in $\mathcal{G}$, we know each edge in $\mathcal{G}$ has at least one end node in $V_C$. By definition, $V_c$ is a vertex cover of size $k$. 

As a result, $P_L$ and $P_{CV}$ is equivalent, proving that minimizing CN metric is NP-hard. The other local metrics are different variations of CN metrics. To make the above proof applicable for other metrics, we need to construct graph $\mathcal{Q}$ such that $f_t(H) = 0$ if and only if there is no common neighbors between each pair of target link. To achieve this, we can slightly modify the graph $\mathcal{Q}$ constructed previously for CN metric. For CND metrics, we can add some isolated nodes to $\mathcal{Q}$ and connect $w$ with each of the isolated nodes. For WCN metrics, we can add some isolated nodes for each node $v_i$ and connect each isolated node with $v_i$ to make sure that the degree of each $v_i$ is always positive. Then the previous proof holds for other local metrics. 
\end{proof}

\subsection{Practical Attacks}

Since in general attacking even local metrics is hard, we have two ways of achieving positive results: approximation algorithms and restricted special cases.
We start with the former, and exhibit several tractable special cases thereafter.

To obtain an approximation algorithm for the general case, we use submodular relaxation.
Specifically, we bound the denominator of each term of $f_t$ by constants as if all the budget were assigned to decrease that single term, arriving at an upper bound $f_{tu}$ for the original objective $f_t$.

For WCN metrics, let $g_{ij}$ be the denominator of $\mathsf{Sim}(u_i,u_j)$. For each $g_{ij}$, we bound it by $L_{ij} \leq g_{ij} \leq U_{ij}$, where $L_{ij}$ is obtained when $k$ edges are deleted and $U_{ij}$ is obtained when no edge is deleted. Take S\o rensen metric as an example, where $\mathsf{Sim}(u_i,u_j) = \frac{2|N(u_i,u_j)|}{d(u_i)+d(u_j)}$. Then $d_i^0 + d_j^0 - k \leq d(u_i) + d(u_j) \leq d_i^0 + d_j^0$, where $d_i^0$ and $d_j^0$ denote the original degrees of $u_i$ and $u_j$, respectively. In this way, each similarity is bounded as
\begin{align*}
\frac{|N(u_i,u_j)|}{U_{ij}} \leq \mathsf{Sim}(u_i,u_j) \leq \frac{|N(u_i,u_j)|}{L_{ij}}.
\end{align*}
Let $f^{WCN}_{tu} = \sum_{ij} \frac{|N(u_i,u_j)|}{L_{ij}}$ and $f^{WCN}_{tl} = \sum_{ij} \frac{|N(u_i,u_j)|}{U_{ij}}$. Then $f^{WCN}_{tl} \leq f^{WCN}_{t} \leq f^{WCN}_{tu}$. 

Similarly, for CND metrics, the denominator in each term $f_r(S_r)$ is bounded by $f_r(S^0_r) - k \leq f_r(S_r) \leq f_r(S^0_r)$, where $S^0_r$ denotes the sum of the $r$th row of the original decision matrix $X^0$. Then $f^{CND}_{tl} \leq f_t^{CND} \leq f^{CND}_{tu}$, where $f^{CND}_{tl} = \sum_{r=1}^m W_r \frac{\sum{ij} x_{ri}x_{rj}}{f_r(S_r)}$ and $f^{CND}_{tu} = \sum_{r=1}^m W_r \frac{\sum{ij}x_{ri}x_{rj}}{f_r(S_r) - k}$. Due to the similarity between the structures of $f_t^{WCN}$ and $f_t^{CND}$, we will focus on $f_t^{WCN}$ and omit the superscript $WCN$ in the following analysis. The proposed approximation algorithm the associated bound analysis are also applicable for $f_t^{CND}$.

\paragraph{Optimizing Bounding Function}
We now consider minimizing $f_{tu}$. Let $S'$ be the set of edges that the attacker chooses to delete. Then set $S'$ is associated with a decision matrix $X'$. For any $S\subset S'$, we have $X \geq X'$, where $X$ is the matrix associated with $S$ and $\geq$ denotes component-wise comparison. Define a set function $F(S) = f_{tu}(X^0) - f_{tu}(X)$. Clearly, $F(\emptyset) = 0$. Then minimizing $f_{tup}$ is equivalent to 
\begin{align}
\label{OPT-WCN-UP-sub}
\max_{S\subset E_Q}\  F(S), \quad \text{s.t.}\ |S|\leq k.
\end{align}

\begin{theorem}
	$F(S)$ is a monotone increasing submodular function.
\end{theorem}
\begin{proof}
	Assume $S\subset S'$, we need to show $F(S) \leq F(S')$. It is equivalent to show $f_{tu}(X) \geq f_{tu}(X')$. Let $C_i$ be the $i$th column of $X$. Then $|N(u_i,u_j)| = \langle C_i,C_j\rangle$, where $\langle C_i,C_j\rangle$ denotes their inner product. Now, $f_{tu}(X) = \sum_{ij} \frac{w_{ij}\langle C_i,C_j\rangle}{L_{ij}}$, where the weights $w_{ij}$ and $L_{ij}$ are constants. Since $X \geq X'$, we have $\langle C_i,C_j\rangle \geq \langle C_i',C_j'\rangle$ for every pair of $i,j$. Thus, $f_{tu}(X) \geq f_{tu}(X')$. That is, $F(S)$ is monotone increasing.
	
	Let an edge $e \notin S'$ be associated with the $p$-th row and $q$-th column entry in $X$. Let $e \cup S$ be associated with a matrix $X^e$, where the only difference between $X^e$ and $X$ is that $x^e_{pq} =0$ while $x_{pq} =1$. Similarly, let $e\cup S'$ be associated with a matrix $X^{'e}$. Define $\Delta (e|S) = F(e \cup S) - F(S)$ and $\Delta(e|S') = F(e \cup S') - F(S')$. Then we need to show $\Delta(e|S) \geq \Delta(e|S')$.
	
	\begin{align*}
	\Delta(e|S) &= f_{tu}(X) - f_{tu}(X^e) = \sum_{j} \frac{w_{jq}}{L_{jq}}\langle C_j,C_q \rangle - \sum_{j} \frac{w_{jq}}{L_{jq}}\langle C^e_j,C^e_q \rangle\\
	& =  \sum_{j} \frac{w_{jq}}{L_{jq}}x_{pj}\cdot x_{pq} - \sum_{j} \frac{w_{jq}}{L_{jq}}x^e_{pj}\cdot x^e_{pq} =  \sum_j \frac{w_{jq}}{L_{jq}} x_{pj},
	\end{align*}
	where the sum $\sum_j$ is over all pairs of $(j,q)$ such that $(u_j,u_q) \in H$. The second equality holds as deleting edge $e$ will only affect the $q$-th column. The last equality holds since $x_{pq} =1$ and $x^e_{pq} = 0$.
	
	Similarly, we can obtain $\Delta(e|S') =  \sum_{j} \frac{w_{jq}}{L_{jq}} x'_{pj}$. Then $\Delta(e|S) - \Delta(e|S') = \sum_{j}\frac{w_{jq}}{L_{jq}} (x_{pj}-x'_{pj})$. Since $(x_{pj} - x'_{pj}) \geq 0$, we have $\Delta(e|S) - \Delta(e|S') \geq 0$. By definition, $F(S)$ is submodular.
\end{proof}

Problem (\ref{OPT-WCN-UP-sub}) is to maximize a monotone increasing submodular function under cardinality constraint. The typical greedy algorithm for such type of problems achieves a $(1-1/e)$-approximation of the maximum.  In particular, the greedy algorithm will delete the edge that will cause the largest increase in $F(S)$ step by step until $k$ edges are deleted. Suppose the greedy algorithm outputs a sub-optimal set $S^*$, which corresponds to a minimizer $X_u^*$ of $f_{tu}(X)$.  We then take the value $f_t(X_u^*)$ as the approximation of $f_t(X^*)$, where $X^*$ is the optimal minimizer of $f_t$. We term this approximation algorithm as \emph{Approx-Local}.

\paragraph{Bound Analysis}
We theoretically analyze the performance of our proposed approximation algorithm \emph{Approx-Local}.\footnote{We note that for the CN metric in particular, the set function $F(S)$ is the actual objective. Consequently, the greedy algorithm above yields a $(1-1/e)$-approximation in this case.}
Let $X^*$, $X_{u}^*$, and $X_{l}^*$ be the minimizers of $f_t$, $f_{tu}$, and $f_{tl}$, respectively. Define  the \emph{gap} between $f_t$ and $f_{tu}$ as $\alpha(X) = f_{tu}(X) - f_t(X)$, which is a function of the decision matrix $X$. 

\begin{theorem}
	\label{thm-gap}
	The gap $\alpha(X)$ is an increasing function of $X$.
\end{theorem}

\begin{proof}
	Consider a particular term of $\alpha (X)$, which is denoted as  $\alpha_{ij}(X) = \frac{w_{ij}}{L_{ij}}\langle C^{X}_i,C^{X}_j\rangle - \frac{w_{ij}}{g(d(u_i),d(u_j),\langle C^{X}_i,C^{X}_j\rangle)}\langle C^{X}_i,C^{X}_j\rangle$, where $C^X_i$ denotes the $i$th column of $X$. For simplicity, write $g(d(u_i),d(u_j),$\\
	$\langle C^{X}_i,C^{X}_j\rangle)$ as $g(X)$.
	
	Consider an edge connecting to $u_i$ is deleted. This corresponds to the case when an entry in $C^X_i$ is erased. Denote the resulting matrix as $Y$. Then $X \geq Y$. The gap at $Y$ is $\alpha_{ij}(Y) = w_{ij}(\frac{\langle C^{Y}_i,C^{Y}_j\rangle}{L_{ij}} - \frac{\langle C^{Y}_i,C^{Y}_j\rangle}{g(Y)})$.
	\begin{align*}
	\frac{\alpha_{ij}(X) - \alpha_{ij}(Y)}{w_{ij}} &= \frac{\langle C^{X}_i,C^{X}_j\rangle - \langle C^{Y}_i,C^{Y}_j\rangle}{L_{ij}} + \frac{\langle C^{Y}_i,C^{Y}_j\rangle}{g(Y)} -\frac{\langle C^{X}_i,C^{X}_j\rangle}{g(X)}
	\end{align*}

	As $g$ is strictly increasing in $d(u_i)$ and $d(u_j)$, it is  increasing in $X$. Then we have $g(X) \geq g(Y)$. Thus,
	\begin{align*}
	\frac{\alpha_{ij}(X) - \alpha_{ij}(Y)}{w_{ij}} &\geq \frac{\langle C^{X}_i,C^{X}_j\rangle - \langle C^{Y}_i,C^{Y}_j\rangle}{L_{ij}} + \frac{\langle C^{Y}_i,C^{Y}_j\rangle}{g(Y)} -\frac{\langle C^{X}_i,C^{X}_j\rangle}{g(Y)}\\
	&= (\langle C^{X}_i,C^{X}_j\rangle - \langle C^{Y}_i,C^{Y}_j\rangle)(\frac{1}{L_{ij}} -\frac{1}{g(Y)}) \geq 0.
	\end{align*}
	The last inequality holds as  $L_{ij}$ is the lower bound (i.e., $L_{ij} \leq g(Y)$). As $\alpha(X)$ is the weighted sum over all pair of target links, we have  $\alpha(X) \geq \alpha(Y)$.
\end{proof}

Theorem \ref{thm-gap} states that the gap between the total similarity and its upper bound function is closing as we delete more edges (i.e., $X$ becomes smaller). We further provide a \emph{solution-dependent} bound of $\mathsf{g}  = f_t(X_{u}^*)- f_t(X^*)$, which measures the gap between the minimum of $f_t$ output by our proposed algorithm and the real minimum. 
\begin{align*}
\mathsf{g} \leq  f_{tu}(X_{u}^*)- f_t(X^*) \leq f_{tu}(X_u^*)- f_{tl}(X^*)\leq f_{tu}(X_u^*)- f_{tl}(X_l^*).
\end{align*}
Such a gap depends on the solutions $X_u^*$ and $X_l^*$.
We evaluate the gap through extensive experiments in Section \ref{Sec-exp}.

\subsection{Tractable Special Cases}\label{Sec-special-case}
We identify two important special cases for which the attack models are significantly simplified. The first case considers attacking a single target link and optimal attacks can be found in linear time for \emph{all} local metrics. The second case considers attacking a group of \emph{nodes} and the goal is to hide all possible links among them. We demonstrate that optimal attacks in this case can be found efficiently for the class of CND metrics.

Due to the space limit, we only highlight some key observations and present some important results. 
The full analysis is in the extended version \cite{extended} of the paper.

\subsubsection{Attacking a Single Link}
When the target is a single link $(u,v)$, the attacker will focus only on the links connecting $u$ or $v$ with their common neighbors, denoted as $N(u,v) = \{w_i\}_{i=1}^s$. Let $x_{iu} = 0$ denotes that attacker chooses to delete the link between $w_i$ and $u$ and $x_{iu} = 1$ otherwise.
\begin{proposition}
	For CND metrics, $\mathsf{Sim}(u,v) = \sum_{i=1}^s \frac{x_{iu}x_{iv}}{g(d(w_i))}$, where $g$ is a non-decreasing function of $d(w_i)$.
\end{proposition}

To minimize a CND, the attacker will remove
edges incident to common neighbors $w$ in increasing order of degree
$d(w)$. In fact, this algorithm is optimal and has a time complexity $\mathcal{O}(|N(u,v)|)$.

For WCN metrics, consider a tuple $(u,w,v)$ where $w$ is a common neighbor of $u$ and $v$. We divide the links surrounding $(u,v)$ into four sets: $E_1 =\{(u,w)\}$, $E_2 = \{(v,w)\}$, $E_3 = \{(u,s)\}$, and $E_4 = \{(v,s)\}$, where $s$ denotes a non-common neighbor of $u$ and $v$. As the attacker deletes links from $E_Q$, there are four possible states of the tuples between $u$ and $v$. In state 1, both $(u,w)$ and $(w,v)$ are deleted. In state 2, $(u,w)$ is deleted while $(w,v)$ is not. In state 3, $(w,v)$ is deleted while $(u,w)$ is not. In state 4, neither $(u,w)$ not $(w,v)$ is deleted. We use integer variables $y_1,y_2,y_3$ to denote the number of tuples in state 1, 2, 3, respectively. Furthermore, let $y_4$ and $y_5$ be the number of deleted edges from $E_3$ and $E_4$, respectively. In this way, the vector $(y_1,y_2,y_3,y_4,y_5)$ fully captures an attacker's  strategy.
\begin{proposition}
	A WCN metric can be written as $\mathsf{Sim}(u,v) = f(y_1,y_2,y_3,y_4,y_5)$ such that  $f$ is decreasing in $y_2$ and $y_3$ and $f$ is increasing in $y_4$ and $y_5$.
\end{proposition}	

Our analysis shows that in an optimal attack, $y_1^* = y_4^* = y_5^* =0$ and $y_2^* + y_3^* = k$. That is, the attacker will always choose $k$ edges from $E_1 \cup E_2$ to delete. The following theorem then specifies how the attacker can optimally choose edges.
\begin{theorem}\label{thm-WCN-single}
	The optimal attack on WCN metrics with a single target link selects arbitrary $y_2^*$ links from $E_1$ and $(k-y_2^*)$ links from $E_2$ to delete with the constraint that for any selected links $(u,w_1) \in E_1$ and $(v,w_2) \in E_2$, $w_1 \neq w_2$. The value of $y_2^*$ is the solution of a single-variable integer optimization problem.
\end{theorem}
The time complexity of solving the single-variable integer optimization problem is bounded in $\mathcal{O}(k)$.

\subsubsection{Attacking A Group of Nodes}
We consider the special case where 1) the target is a group of nodes $U$ and the links between each pair of nodes in $U$ consist the target link set $H$; 2) each link in $H$ has equal weight. In this case, optimal attacks on CND metrics can be found in polynomial time. 

\begin{proposition}
	\label{Prop-CND}
	For CND metrics, the total similarity $f_t$ has the form $\sum_{i=1}^m f_i(S_i)$, where $S_i$ is the sum of the $i$th row of $X$ and $f_i(S_i)$ is a convex increasing function of $S_i$.
\end{proposition}

Proposition \ref{Prop-CND} states that $f_t$ for CND metrics can be written as a sum of independent functions, where each function $f_i$ is a convex increasing function.  We then propose a greedy algorithm termed \emph{Greedy-CND} to minimize $f_t^{CND}$. In essence, \emph{Greedy-CND} takes as the input $\mathbf{S}^0$, which is the row sum of the initial decision matrix $X$, and decreases an entry in $\mathbf{S}^0$ whose decreasing  causes the maximum decrease in $f_t^{CND}$ step by step until an upper bound of $k$ edges are deleted. 
This algorithm turns out to be optimal, as we prove in the extended version of the paper.

%% file: global.tex
\section{Attacking Global Metrics}

In this section, we analyze attacks on two common global similarity metrics: Katz and ACT. We begin with attacks on a single link and show that finding optimal attack strategies is NP-hard even for a single target link.

Let $A \in \{0,1\}^{N\times N}$  and $D$ be the adjacency matrix and degree matrix of the graph  $\mathcal{G}_Q$, respectively.  The Laplacian matrix is defined as $L = D-A$. The pseudo-inverse of $L$ is  $L^{\dag} = (L -E)^{-1} +E$, where $E$ is an $(N \times N)$ matrix with each entry being $\frac{1}{N}$. We use a binary vector $\mathbf{y} \in \{0,1\}^M$ to denote the states of edges in $E_Q$, where $y_i = 0$ iff the $i$th edge in $E_Q$ is deleted. 
Finally, $\mathbf{y} \leq \mathbf{y}'$ ($A\leq A'$) is a component-wise inequality between vectors (matrices). 

\subsection{Problem Formulation for Katz Similarity}

The Katz similarity is a common path-based similarity metric~\cite{katz1953new}.
For a pair of nodes $(u,v)$, Katz similarity is defined as 
$$
\mathsf{Katz}(u,v) = \sum_{l = 1}^\infty\beta^l |path^l_{u,v}| = (\beta A + \beta^2 A^2 + \beta^3 A^3 + \cdots)_{uv},
$$
where $|path^l_{u,v}|$ denotes the number of walks of length $l$ between $u$ and $v$, $\beta > 0$ is a parameter and $(\cdot)_{uv}$ denotes the entry in the $u$th row and $v$th column of a matrix. 
By definition, the adjacency matrix $A$ is fully captured by the vector $\mathbf{y}$. Thus, $\mathsf{Katz}(u,v)$ is a function of $\mathbf{y}$, written as $\mathsf{Katz}_{uv}(\mathbf{y})$. 
As one would expect, it is an increasing function of $\mathbf{y}$.
\begin{lemma}\label{theorem-Katz}
	$\mathsf{Katz}_{uv}(\mathbf{y})$ is an increasing function of $\mathbf{y}$. 
\end{lemma}
\begin{proof}
	Let $A$ and $A'$ be the corresponding adjacency matrices of $\mathbf{y}$ and $\mathbf{y}'$. If $\mathbf{y} \leq \mathbf{y}'$, we have $A \leq A'$. Now, consider the $j$th term of the Katz similarity matrix $K$, which is $\beta^j A^j$. As every entry in $A$ is non-negative and $\beta > 0$, we have $\beta^j A^j \leq \beta^j A'^j$, for every $j$.  Thus, $\mathsf{Katz}_{uv}(\mathbf{y}) \leq  \mathsf{Katz}_{uv}(\mathbf{y}')$.
\end{proof}
As a result, deleting a link will always decrease $\mathsf{Katz}_{uv}(\mathbf{y})$, and the attacker would therefore always delete $k$ links in $E_Q$ (if $E_Q$ has at least $k$ links). 
Thus, minimizing Katz for a particular target link $(u,v)$ is captured by  \emph{Prob-Katz}:
$$
\min_\mathbf{y} \ \mathsf{Katz}_{uv}(\mathbf{y}), \quad \text{s.t.} \quad \sum_{i=1}^M y_i = M-k, \mathbf{y} \in \{0,1\}^M.
$$

\subsection{Problem Formulation for ACT}

The second global similarity metric we consider is based on ACT, which measures a distance between two nodes in terms of random walks.
Specifically, for a pair of nodes $(u,v)$, $\mathsf{ACT}(u,v)$, is the expected time for a simple random walker to travel from a node $u$ to node $v$ on a graph and return to $u$. Since $\mathsf{ACT}(u,v)$ is a distance metric, the attacker's aim is to \emph{maximize} $\mathsf{ACT}(u,v)$, defined as
$$\mathsf{ACT}(u,v) = V_G (L^\dag_{uu} + L^\dag_{vv} - 2L^\dag_{uv}),$$
where $V_G$ is the volume of the graph \cite{fouss2007random}.

Directly optimizing $\mathsf{ACT}(u,v)$ is hard. Indeed, deleting an edge may either increase or decrease $\mathsf{ACT}(u,v)$, so that unlike other metrics, ACT is not monotone in $\mathbf{y}$. 
Fortunately, \citeauthor{ghosh2008minimizing} \shortcite{ghosh2008minimizing} show that when edges are unweighted (as in our setting), $\mathsf{ACT}(u,v)$ can be defined in terms of \emph{Effective Resistance (ER)}: 
$\mathsf{ACT}(u,v) = V_G \mathsf{ER}(u,v).$
It is also not difficult to see that both the volume $V_G$ and ER can be represented in terms of $\mathbf{y}$.

We begin by investigating the effect of deleting an edge on $\mathsf{ER}(\mathbf{y})$.  
We use a well-known result by \citeauthor{doyle2000random}~\shortcite{doyle2000random} to this end.
\begin{lemma}[\cite{doyle2000random}]\label{lemma-ER}
	The effective resistance between two nodes is  strictly increasing when an edge is deleted. 
\end{lemma}

The following lemma is then an immediate corollary.
\begin{lemma}\label{theorem-ER}
	$\mathsf{ER}(\mathbf{y})$ is a decreasing function of $\mathbf{y}$. 
\end{lemma}

As a result, maximizing $\mathsf{ER}(\mathbf{y})$ would always entail deleting all allowed edges.
Let $t$ be the maximum number of edges that can be deleted.
Then, maximizing $\mathsf{ER}(\mathbf{y})$ can be formulated as \emph{Prob-ER}:
\begin{align*}
\max_\mathbf{y} \ \mathsf{ER}(\mathbf{y}), \quad \text{s.t.} \ \sum_{i=1}^M y_i = M-t, \mathbf{y} \in \{0,1\}^M.
\end{align*}
However, while $\mathsf{ER}(\mathbf{y})$ increases as we delete edges, volume $V_G = 2\sum_{i=1}^M y_i$ decreases.
Fortunately, since volume is linear in the number of deleted edges, we reduce the problem of optimizing $\mathsf{ACT}$ to that of solving \emph{Prob-ER} by solving the latter for $t = \{0,\ldots,k\}$, and choosing the best of these in terms of $\mathsf{ACT}$.
Similarly, hardness of \emph{Prob-ER} implies hardness of optimizing $\mathsf{ACT}$.
Consequently, the rest of this section focuses on solving \emph{Prob-ER}.

\subsection{Hardness Results}
We prove that minimizing Katz and maximizing ER between a single pair of nodes by deleting edges with budge constraint are both NP-hard. 

\begin{theorem}
	\label{thm-global}
	Minimizing Katz similarity and maximizing ACT distance is NP-hard even if $H$ contains a single target link.
\end{theorem}
\begin{proof}
We consider the decision version of minimizing Katz, termed as $P_{K}$, which is to decide whether one can delete $k$ edges to make  $\mathsf{Katz}(u,v) \leq q$ given a graph $\mathcal{Q}$ and a target node pair $(u,v)$ in $\mathcal{Q}$. Similarly, we consider the decision version of maximizing ER, termed as $P_{E}$: which is to decide whether one can delete $k$ edges to make  $\mathsf{ER}(u,v) \geq q$ given a graph $\mathcal{Q}$ and a target node pair $(u,v)$ in $\mathcal{Q}$.

We use the Hamiltonian cycle problem, termed $P_H$, for reduction. $P_H$ is to decide whether there exists a circle that visits each nodes in a given connected graph $G$ exactly once (thus called Hamiltonian circle).

Before reduction, we first consider the minimum of Katz and maximum of ER between two nodes $u$ and $v$ in all possible connected graphs over a fixed node set. By the definition of Katz similarity, $\mathsf{Katz}(u,v)$ is minimized when the graph is a \emph{string} with $u$ and $v$ as two end nodes and all others as inner nodes in that string; that is the graph over that set of nodes is a Hamiltonian path with $u$ and $v$ as end nodes. We denote the minimum value of $\mathsf{Katz}(u,v)$ in this case as $min_K$. Similarity, by the definition of effective resistance, $\mathsf{ER}(u,v)$ is maximized when the graph is also a Hamiltonian path over that set of nodes with  $u$ and $v$ as the two end nodes. We assume that all edges have equal resistance. We denote the maximum value of $\mathsf{ER}(u,v)$ in this case as $max_E$. 

We then set $q=min_K$ in the decision problem $P_E$ and set $q= max_E$ in $P_E$. As a result, the two decision problems $P_E$ and $P_K$ are then both equivalent to the following decision problem, termed $P_S$: given a graph $\mathcal{Q}$ and two nodes $u$ and $v$ in $\mathcal{Q}$, can we delete $k$ edges such that the remaining graph $\mathcal{S}$ forms a string (i.e., a Hamiltonian path) with  $u$ and $v$ as two end nodes?

Now the reduction. Given an instance of Hamiltonian circle (i.e., a graph $\mathcal{G}=(V,E)$), we construct a new graph $\mathcal{Q}$ from $\mathcal{G}$ in the following steps:
\begin{itemize}
	\item Select an arbitrary node $w$ in $\mathcal{G}$. Let $N(w) = \{l_1,l_2,\cdots,l_W\}$ be the neighbors of $w$, where $W= |N(w)|$.
	\item Add two nodes $u$ and $v$. 
	\item Add edge $(u,w)$ and edges $(v,l_i)$, $\forall l_i \in N(w)$.
\end{itemize}
The resulting graph is then the graph $\mathcal{Q}$ in decision problem $P_S$, where the budget $k = W + |E| -|V|$. Below is an example showing the construction of $\mathcal{Q}$ and the process of deleting edges. 

\begin{figure}[ht]
	\centering
	\includegraphics[scale=0.2]{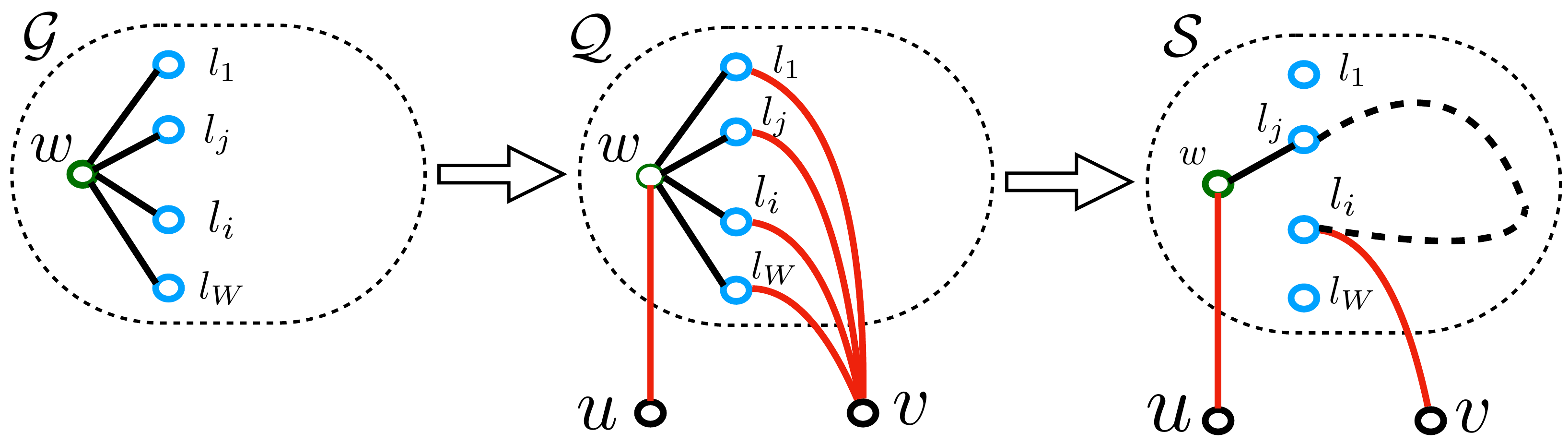}
	\caption{Illustration: construct graph $\mathcal{Q}$ from $\mathcal{G}$ and delete edges from $\mathcal{Q}$ to obtain $\mathcal{S}$. }
	\label{fig-plot}
\end{figure}

We now show that problem $P_H$ and problem $P_S$ are equivalent. 

First, we show if there exists a Hamiltonian circle in $\mathcal{G}$, then we can delete $k = W + |E| -|V|$ edges such that the measurement (Katz or ER) between $u$ and $v$ in graph $\mathcal{Q}$ is $q$. Assume the Hamiltonian circle travels to $w$ through edge  $(l_i,w)$ and leaves $w$ through edge $(w,l_j)$. We then 1) delete $(W-1)$ edges $(v,l_t)$ for each $l_t \in N(W)$ and $l_t \neq l_i$; 2) delete all $|E| -|V|$ edges in $G$ that do not appear in the Hamiltonian circle; 3) delete edge $(w,l_i)$. Thus, we deleted a total of $W + |E| -|V|$ edges. After deleting all these $k$ edges, in the remaining graph $\mathcal{S}$, there exists a Hamiltonian path between $w$ and $l_i$. As $u$ only connects to $w$ and $v$ only connects to $l_i$, the remaining graph forms a  Hamiltonian path between $u$ and $v$. As a result, the measurement between $u$ and $v$ equals $q$.

Second, we show if we can remove $k =  W + |E| -|V|$ edges from $\mathcal{Q}$ such that the remaining graph $\mathcal{S}$ forms a Hamiltonian path between $u$ and $v$, then we can find a Hamiltonian circle in the graph $\mathcal{G}$. Suppose in the reaming string, $v$ connects to $l_i$ and $w$ connects to $l_j$. As $u$ connects only to $w$ in graph $\mathcal{Q}$, $u$ must connect to $w$ in $\mathcal{S}$. From the construction of $\mathcal{Q}$, the total number of edges of $\mathcal{Q}$ is $|E| +W +1$. After deleting $k$ edges, the remaining number of edges is $|V|+1$. Excluding the two edges $(u,w)$ and $(v,l_i)$, we know there are $|V|-1$ edges among the node set $V$ of the original graph $\mathcal{G}$. As the remaining graph is connected, there must exist a Hamiltonian path between $w$ and $l_i$. As $(w,l_i)$ is an edge in the graph $\mathcal{G}$, we have found a Hamiltonian circle in $\mathcal{G}$, consisting of the Hamiltonian path between $w$ and $l_i$ plus the edge $(w,l_i)$.

Thus, decision problem $P_S$ is NP-complete; minimizing Katz and maximizing ER (ACT) are NP-hard.
\end{proof}

\subsection{Practical Attack Strategies}\label{Sec-Practical-Attacks}
While computing an optimal attack on Katz and ACT is NP-Hard, we now devise approximate approaches which are highly effective in practice.

\subsubsection{Attacking Katz Similarity}
To attack Katz similarity, we transform the attacker's optimization problem into that of maximizing a monotone increasing submodular function.
We begin with the single-link case (i.e., $H$ is a singleton), and subsequently generalize to an arbitrary $H$. 
We define a set function $f(S_p)$ as follows.
Let $S_p \subseteq E_Q$ be a set of edges that an attacker chooses to delete. Let $A_p$ be the adjacency matrix of the graph $\mathcal{G}_Q$ after all the edges in $S_p$ are deleted. Define 
$$f(S_p) = \beta A_p + \beta^2 A_p^2 + \beta^3 A_p^3 + \cdots$$

Since there is a one-to-one mapping between the set $S_p$ and the matrix $A_p$, the function $f(S_p)$ is well-defined. We note that $f(S_p)$ gives the Katz similarity matrix of the graph $\mathcal{G}$ after all the edges in $S_p$ are deleted. We further define a set function
$$g_{uv}(S_p) = (K - f(S_p))_{uv},$$
where $K = f(\emptyset)$ (the Katz similarity matrix when no edges are deleted) and  $(\cdot)_{uv}$ denotes the $u$th row and $v$th column of a matrix. 
Clearly, when $S_p = \emptyset$, $g_{uv}(S_p) = 0$.

Then, \emph{Prob-Katz} is equivalent to 
\begin{align}
\label{OPT-5}
\max_{S_p\subset E_t} \  g_{uv}(S_p), \quad \text{s.t.} \ |S_p| = k
\end{align}

\begin{theorem}
	\label{Theorem-Katz}
	The set function $g_{uv}(S_p)$ is monotone increasing and submodular.
\end{theorem}

\begin{proof}
	To prove that $g_{uv}$ is monotone increasing, we need to show that $\forall S_p \subset S_q \subset Q$, $g_{uv}(S_p) \leq g_{uv}(S_q)$. It is equivalent to show $(f(S_p))_{uv} \geq (f(S_q))_{uv}$. We note that $(f(S_p))_{uv}$ and $(f(S_q))_{uv}$ are the Katz similarity between $u$ and $v$ after the edges in $S_p$ and $S_q$ are deleted, respectively. Theorem \ref{theorem-Katz} states that the Katz similarity will decrease as more edges are deleted. Since $S_p \subset S_q$, we have $f(S_p) \geq f(S_q)$. Thus, $g_{uv}(S_p) \leq g_{uv}(S_q)$.
	
	Next, we prove $g_{uv}$ is submodular. Let $e \in E_t\setminus S_q$ be an edge between node $i$ and node $j$ in the graph. Let $G$ be an $n\times n$ matrix where $G_{ij} = G_{ji} = 1$ and the rest of the entries are $0$. Then we have the set $S_p \cup e$ is associated with $A_p -G$ and $S_q \cup e$ is associated with $A_q - G$. For a set $S$, let $\Delta (e|S) = f(S\cup e) - f(S)$.  Then we need to show 
	$$\Delta (e|S_p) \leq \Delta(e|S_q).$$
	
	Denote the $t$th item of $\Delta(e|S)$ as $\Delta^{(t)}(e|S)$. In the following, we will first prove $\Delta^{(t)}(e|S_p) \leq \Delta^{(t)}(e|S_q)$ by induction. Assume that the inequality holds for $t = s$ (it's straightforward to verify the case for $t=1$ and $t=2$). That is 
	\begin{eqnarray}
	\beta^s [(A_p - G)^s - (A_p)^s - (A_q - G)^s + (A_q)^s] \leq \mathbf{0}.
	\end{eqnarray}
	When $t = s+1$, we have
	\begin{eqnarray*}
		&&(\Delta^{(s+1)}(e|S_p) - \Delta^{(s+1)}(e|S_q))/\beta^{s+1}\\
		&=&(A_p - G)^{s+1} - (A_p)^{s+1} - (A_q - G)^{s+1} + (A_q)^{s+1}\\
		&=&(A_p - G)^s A_p - (A_p)^{s+1} - (A_q-G)^s A_q + (A_q)^{s+1} \\
		&-& [(A_p-G)^{s+1} - (A_q-G)^{s+1}]G\\
		&\leq& (A_p - G)^s A_p - (A_p)^{s+1} - (A_q-G)^s A_q + (A_q)^{s+1}
	\end{eqnarray*}
	The inequality comes from the fact that $(A_p-G) \geq (A_q -G)$ when $G \geq \mathbf{0}$. Furthermore, since $S_p \subset S_q$, we have $A_p = A_q +F$ for some $F \geq \mathbf{0}$. Thus,
	\begin{eqnarray*}
		&&(\Delta^{(s+1)}(e|S_p) - \Delta^{(s+1)}(e|S_q))/\beta^{s+1}\\
		&\leq&(A_p - G)^s(A_q +F) - (A_p)^s (A_q + F) \\
		&-& (A_q - G)^s A_q + (A_q)^{s+1}\\
		&=&[(A_p - G)^s  - (A_p)^s - (A_q-G)^s + (A_q)^s]A_q \\
		& +& [(A_p - G)^s - (A_p)^s]F\\
		&\leq& \mathbf{0}
	\end{eqnarray*}
	
	By induction, we have $\Delta^{(t)}(e|S_p) \leq \Delta^{(t)}(e|S_q)$ for $t =1,2,3,\cdots$. Note that when $\beta$ is chosen to be less than the reciprocal of the maximum of the eigenvalues of $A_q -G$, the sum will converge. Thus, $\Delta(e|S_p) \leq \Delta(e|S_q)$.
\end{proof}

Next, for the multi-link case, the total similarity $f_t = \sum_{i,j} w_{ij}K_{ij}$. Let $F(S)$ be a function of the set of deleted edges, defined as
$$F(S) = \beta A_S + \beta^2 A_S^2 + \beta^3 A_S^3 + \cdots,$$
where $A_S$ denotes the adjacency matrix after all edges in $S$ are deleted. Note that $F(S)$ gives the Katz similarity matrix when edges in $S$ are deleted. Further define $g_{ij}(S) = (K^0-F(S))_{ij}$, where $K^0$ is the original Katz similarity matrix.  Let $G_t(S) = \sum_{i,j} w_{ij}g_{ij}(S)$. By definition, we have $G_t(S) = \sum_{i,j} w_{ij}K^0_{ij} - f_t$. Thus, minimizing $f_t$ is equivalent to 
$$\max_{S \subset E_Q}\  G_t(S), \quad \text{s.t.}\ |S|\leq k.$$
The following result is then a direct corollary of Theorem~\ref{Theorem-Katz}.
\begin{corollary}
	$G_t(S)$ is monotone increasing and submodular.
\end{corollary}

\begin{proof}
	This is an immediate conclusion of two results. First, $g_{ij}(S)$ is monotone increasing and submodular in $S$ as proved in Theorem \ref{Theorem-Katz}. Second, a positive linear combination of submodular functions is submodular \cite{nemhauser1978analysis}. As $G_t(S)$ is the sum of $g_{ij}(S)$, $G_t(S)$ is monotone increasing and submodular.
\end{proof}

As a result, minimizing the total Katz similarity is equivalent to maximizing a monotone increasing submodular function under cardinality constraint.  We can achieve a $(1-1/e)$ approximation by applying a simple iterative greedy algorithm in which we delete one edge at a time that maximizes the marginal impact on the objective. We call this resulting algorithm \emph{Greedy-Katz}.

\subsubsection{Attacking ACT}
From the analysis of minimizing Katz similarity, it is natural to investigate submodularity of the effective resistance or ACT as a function of the set of edges. Unfortunately, counter examples show that the effective resistance is neither submodular nor supermodular. Consequently, we need to leverage a different kind of structure for ER.

Our first step is to approximate the objective function $\mathsf{ER}(u,v)$ based on the results by \citeauthor{von2014hitting}~\shortcite{von2014hitting}, who show that $\mathsf{ER}(u,v)$ can be approximated by $\frac{1}{d(u)} + \frac{1}{d(v)}$ for large geometric graphs as well as random graphs with given expected degrees.
Consequently, we use the approximation $\mathsf{ER}(u,v) \approx \mathsf{ER}_{ap}(u,v) = \frac{1}{d(u)} + \frac{1}{d(v)}$. Then the total effective resistance is approximated as $\mathsf{ER}(H) \approx \mathsf{ER}_{ap}(H) = \sum_{ij}w_{ij}(\frac{1}{d(u_i)}+\frac{1}{d(u_j)}) = \sum_{i=1}^n \frac{W_i}{d(u_i)}$, where $W_i >0$ is some constant weight associated with each $u_i$. Let $D_i$ be the original degree of node $u_i$ and $z_i$ be an integer variable denoting the number of deleted edges connecting to $u_i$. Then maximizing $\mathsf{ER}_{ap}(H)$ is equivalent to
\begin{align}
\label{OPT-ACT}
\max_{\mathbf{z}}\ \sum_{i=1}^n \frac{W_i}{D_i - z_i}, \quad \text{s.t.} \ \sum_{i=1}^n z_i \leq k, \ z_i\in [0,k].
\end{align}

We assume that deleting edges would not make the graph disconnected. That is $\forall i \in [1,n]$, $k < D_i$.

We formulate the above problem as a linear integer program. Specifically, let $\Delta_{ij}$ be the decrease in $\mathsf{ER}_{ap}(H)$ after $j$ edges connecting to node $u_i$ are deleted. As any such $j$ edges will cause the same decrease, the value of each $\Delta_{ij}$ for $j = 0,1,\cdots,k$ could be efficiently computed in advance. We use a binary variable $h_{ij} = 1$ to denote that the attacker chooses to delete such $j$ edges; otherwise, $h_{ij} = 0$. Then problem (\ref{OPT-ACT}) is equivalent to 
\begin{align*}
\max_{\mathbf{h}}\ \sum_{i=1}^n \sum_{j=0}^k (\frac{W_i}{D_i}- \Delta_{ij}) h_{ij},\ 
\text{s.t.} \ \sum_{i=1}^n \sum_{j=0}^k  h_{ij} \leq k,\forall i,\sum_{j=0}^k h_{ij} \leq 1,
\end{align*}

The above problem is a linear program of $(k+1) \times n$ binary variables with $(n+1)$ linear constraints. A numerical solution \cite{gurobi} gives the number of edges incident to each node that needs to be deleted.

\paragraph{A Note on a Special Case} We can consider problem~(\ref{OPT-ACT}) for the two special cases introduced in Section~\ref{Sec-special-case}. 
For both cases (attacking a single link and attacking a group of nodes), the weights $W_i$ associated with each node are equal. Under this condition, we can prove that the optimal solution to problem~(\ref{OPT-ACT}) is to delete all $k$ edges connecting the node with the smallest degree. Details are provided in the extended version of the paper.

%% file: exp.tex
\section{Experiments}
\label{Sec-exp}
Our experiments use two classes of networks:
1) randomly generated scale-free networks and 2) a Facebook friendship
network \cite{snapnets}. In the scale-free networks, the degree distribution satisfies
$P(k)\propto k^{-\gamma}$, where $\gamma$ is a parameter.

\emph{Baseline algorithms}. We compare our algorithms with two baseline algorithms. We term the first one as \emph{RandomDel}, which randomly deletes the edges connected to the target nodes. The second baseline, termed \emph{GreedyBase}, is a heuristic algorithm proposed in  \cite{waniek2018attack}. This algorithm will try to delete the link whose deletion will cause the largest decrease in the number of ``closed triads" as defined in \cite{waniek2018attack}. Our experiments show that while the performance of \emph{GreedyBase} varies regarding different metrics, \emph{RandomDel} performs poorly for all metrics (Fig.~\ref{fig-local}). Henceforth, we only compare our algorithm with \emph{GreedyBase} for global metrics (Fig.~\ref{fig-katz} and Fig.~\ref{fig-act}).

For local metrics, we evaluate \emph{Approx-Local} in the general
case. We consider a target set of size $20$. We select RA (CND metric)
and Sorensen (WCN metric) as two representatives, for which the
results are presented in Fig.~\ref{fig-local}. All similarity scores
are scaled to $1.0$ when no edges are deleted.  Due to space limit, we
only present the results on one scale of the scale-free network
($n=1000,\gamma = 2.0$) and Facebook network($n=786, m =12291$). 
A more comprehensive set of experiments is presented in the extended version.

We note that deleting a relatively small number of links can significantly decrease the similarities of a set of target links. The gap between the upper and lower bound functions, which reflects the approximation quality of \emph{Approx-Local}, is within $20\%$ of the original similarity.

For global metrics, we evaluate \emph{Greedy-Katz} and
\emph{Local-ACT} regarding a set of target links ($|H| = 20$) on
different scales of networks. As shown in Fig.~\ref{fig-katz} and
Fig.~\ref{fig-act}, the performances are significantly better than
those of the baseline algorithm. Additional results
for the special cases are provided in the extended version.

\begin{figure}[ht]
	\centering
	\begin{subfigure}[t]{0.23\textwidth}
		\centering
		\includegraphics[scale=0.25]{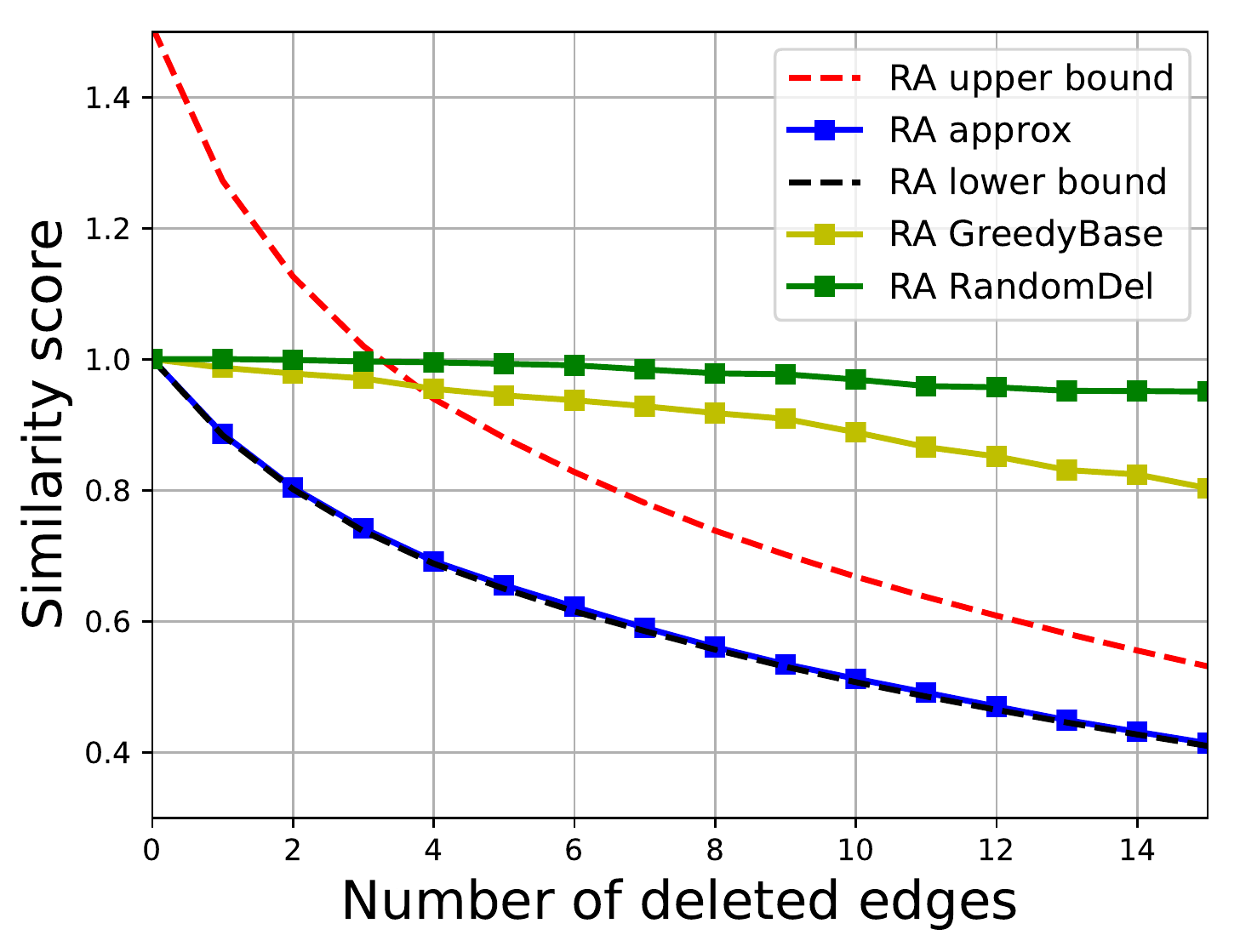}
		\caption{scale-free, RA}
	\end{subfigure}%
	\hfill
	\begin{subfigure}[t]{0.23\textwidth}
		\centering
		\includegraphics[scale=0.25]{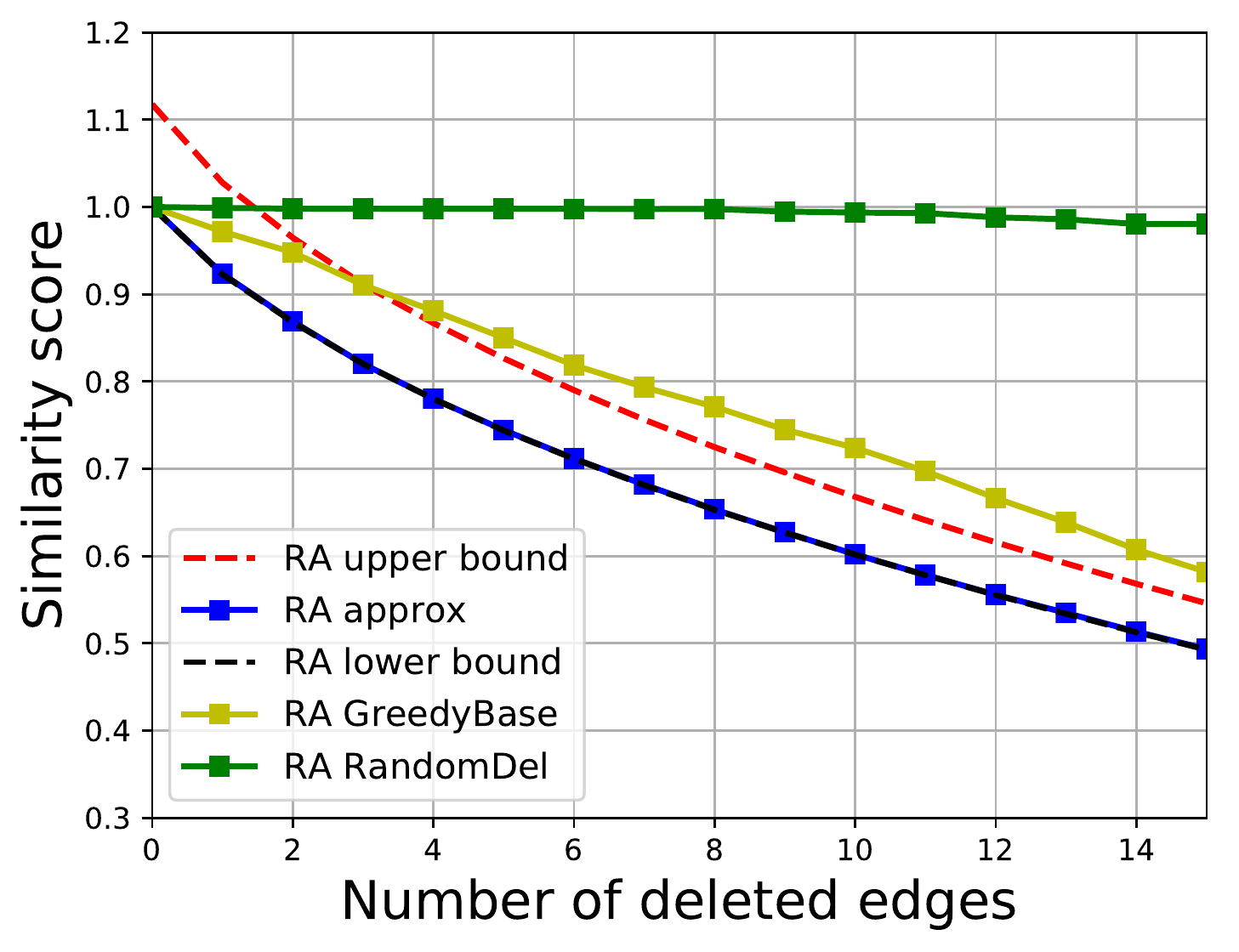}
		\caption{Facebook, RA}
	\end{subfigure}

	\centering
	\begin{subfigure}[t]{0.23\textwidth}
		\centering
		\includegraphics[scale=0.25]{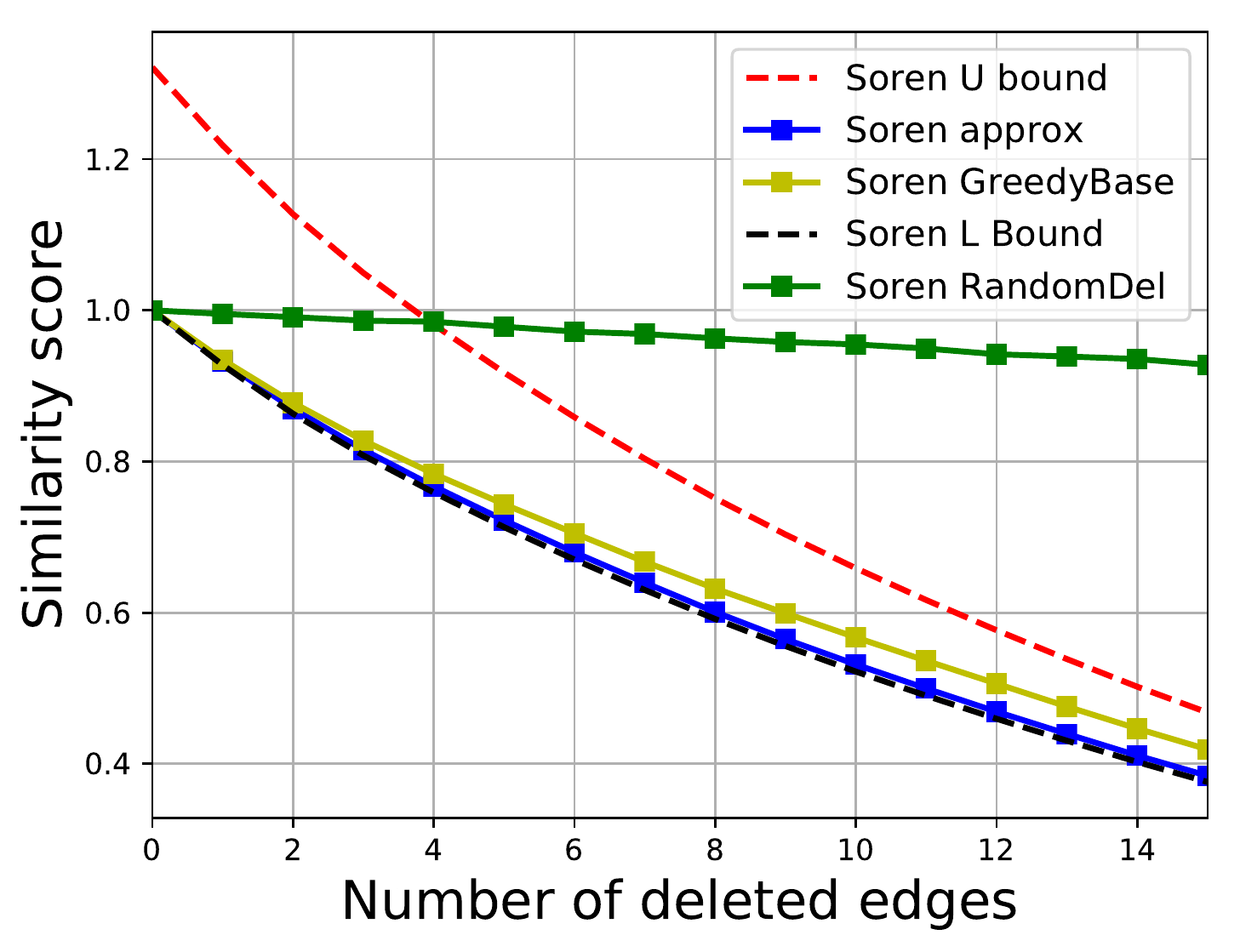}
		\caption{scale-free, S\o rensen}
	\end{subfigure}%
	\hfill
	\begin{subfigure}[t]{0.23\textwidth}
		\centering
		\includegraphics[scale=0.25]{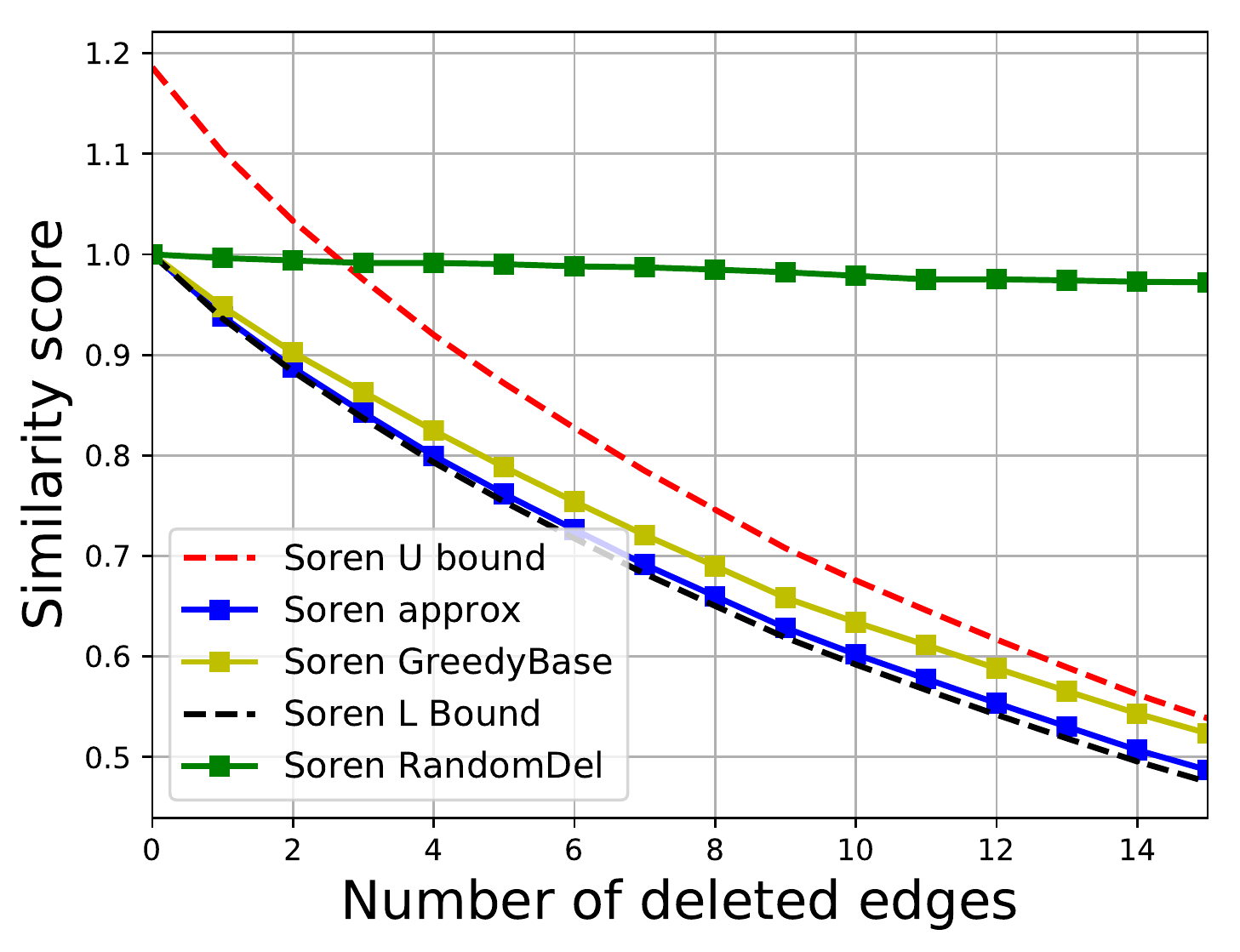}
		\caption{Facebook, S\o rensen}
	\end{subfigure}	
	\caption{Approx-Local vs. GreedyBase on CND (e.g., RA) and WCN (e.g., S\o rensen) metrics in general case.}
	\label{fig-local}
\end{figure}




\begin{figure}[ht]
	\centering
	\begin{subfigure}[t]{0.235\textwidth}
		\centering
		\includegraphics[scale=0.25]{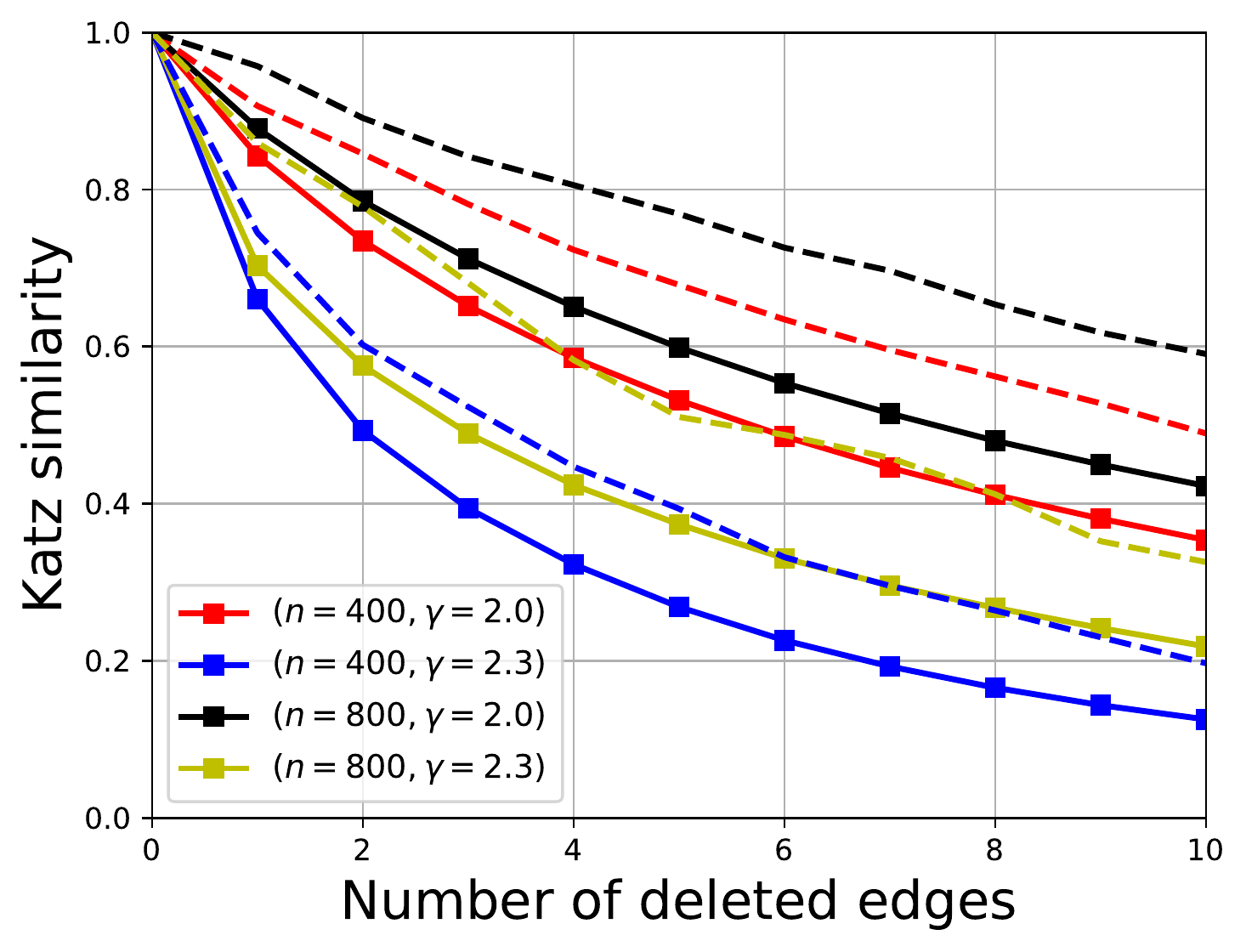}
		\caption{scale-free}
	\end{subfigure}%
	\hfill
	\begin{subfigure}[t]{0.235\textwidth}
		\centering
		\includegraphics[scale=0.25]{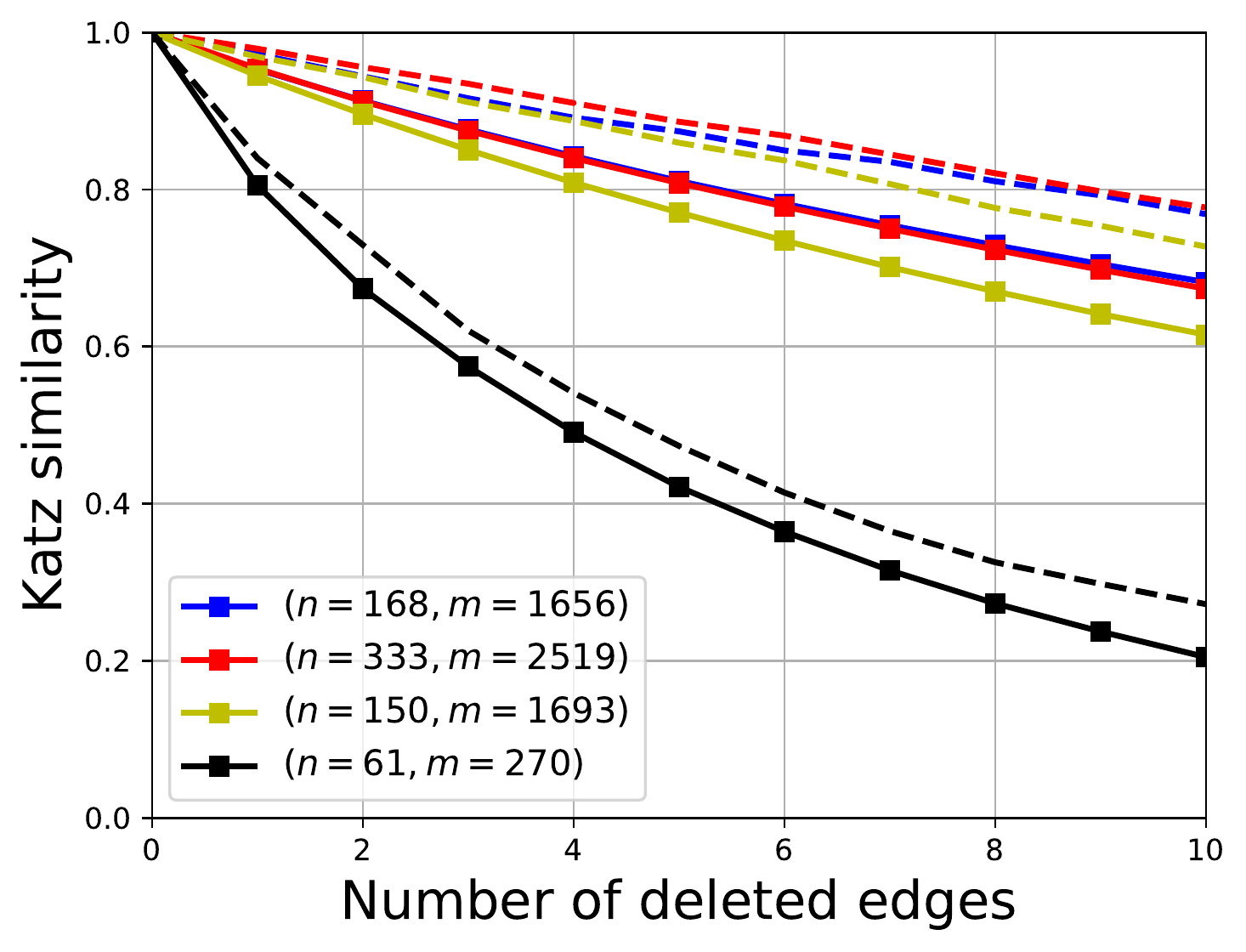}
		\caption{Facebook}
	\end{subfigure}	
	\caption{Greedy-Katz vs. GreedyBase on Katz similarity. Solid lines: Greedy-Katz. Dotted lines: GreedyBase}
	\label{fig-katz}
\end{figure}

\begin{figure}[ht]
	\centering
	\begin{subfigure}[t]{0.235\textwidth}
		\centering
		\includegraphics[scale=0.25]{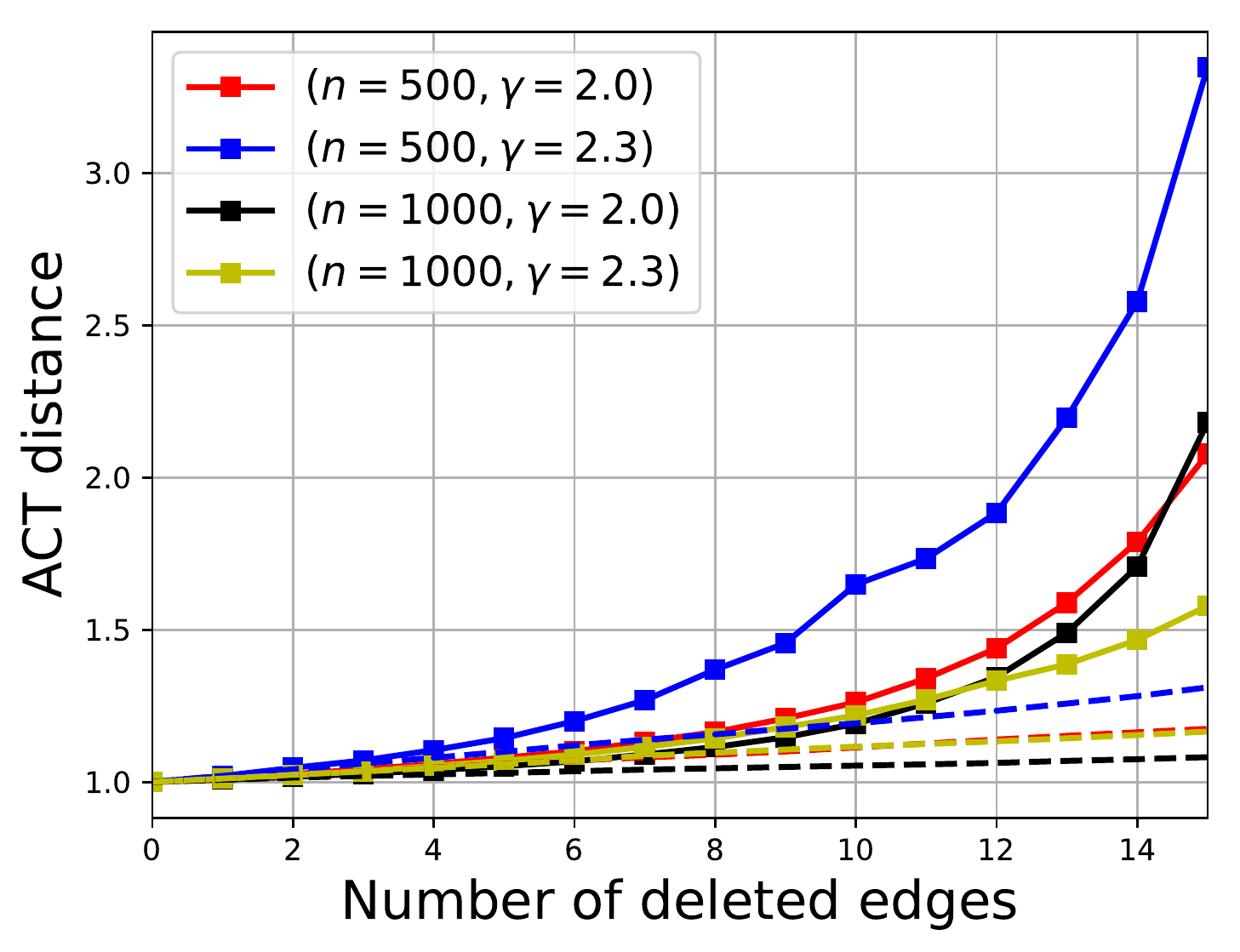}
		\caption{scale-free}
	\end{subfigure}%
	\hfill
	\begin{subfigure}[t]{0.235\textwidth}
		\centering
		\includegraphics[scale=0.25]{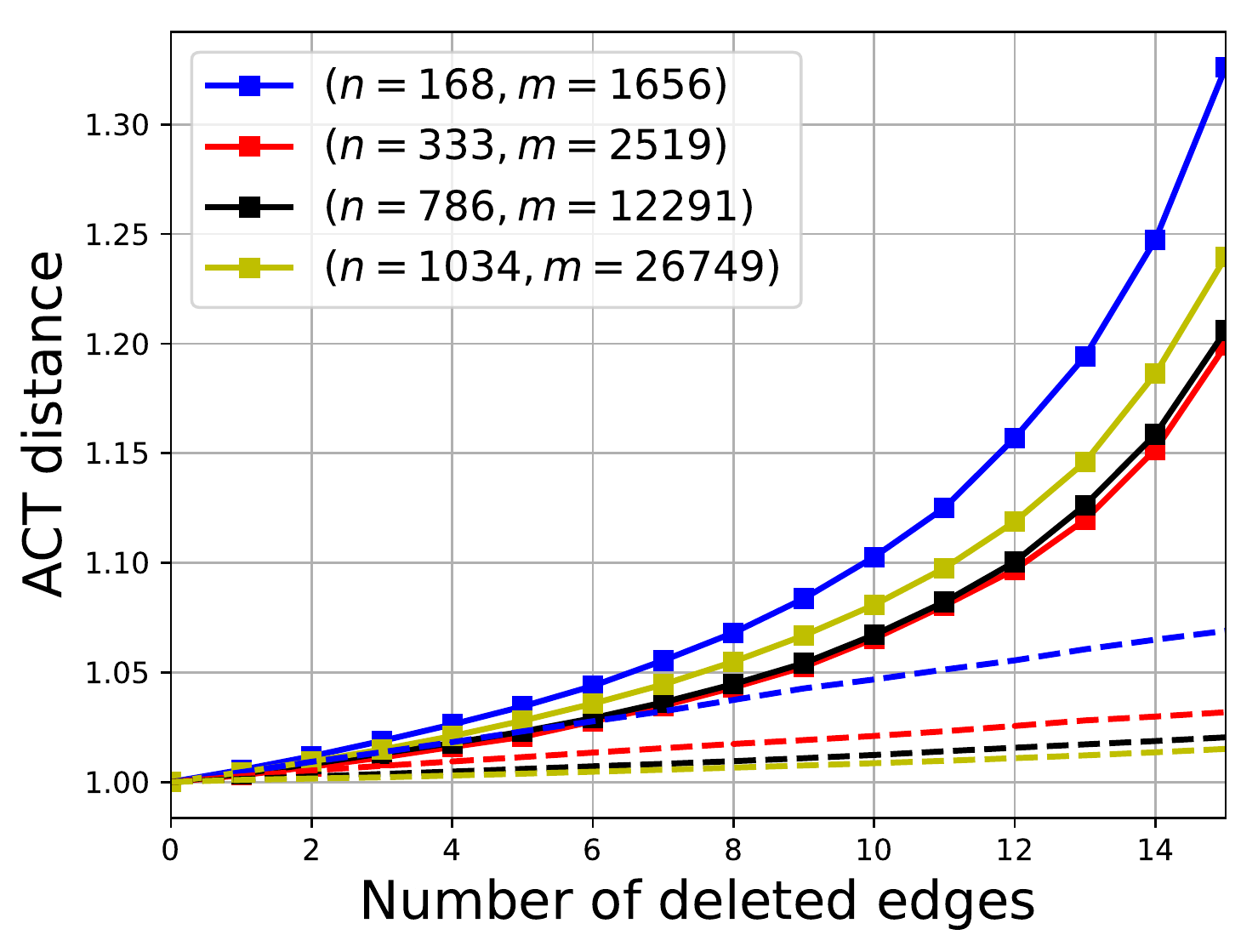}
		\caption{Facebook}
	\end{subfigure}	
	\caption{Local-ACT vs. GreedyBase on ACT distance. Solid lines: Local-ACT. Dotted lines: GreedyBase}
	\label{fig-act}
\end{figure}
